\documentclass[11pt]{article}

\usepackage{fullpage,wrapfig}
\usepackage{amsmath, amsthm, amssymb, algorithm,thmtools,algpseudocode,enumerate,caption,framed}
\usepackage{paralist, sidecap}
\usepackage[usenames,dvipsnames,svgnames,table]{xcolor}
\definecolor{darkgreen}{rgb}{0.0,0,0.9}
\usepackage[colorlinks=true,pdfpagemode=UseNone,citecolor=Blue,linkcolor=Red,urlcolor=Black,pagebackref]{hyperref}
\usepackage{tikz}
\usepackage[T1]{fontenc}
\usepackage[utf8]{inputenc}
\usepackage{authblk}

\renewcommand{\vec}[1]{\mathbf{#1}}

\newtheorem{theorem}{Theorem}[section]

\newtheorem{lemma}{Lemma}[section]

\newtheorem{definition}{Definition}[section]

\newcommand{\bvpg}{\textsc{B$_1$-VPG }}
\newcommand{\bKvpg}{\textsc{B$_k$-VPG }}
\newcommand{\piviK}{$\mathcal{P}$ }
\newcommand{\gr}{\texttt{grid}}
\newcommand{\pat}{\texttt{paths}}

\newcommand{\mds}{\textsc{Minimum Dominating Set }}
\newcommand{\mwis}{\textsc{Maximum-Weighted Independent Set }}
\newcommand{\apx}{\textsc{APX}-hard }
\newcommand{\cor}{\texttt{corner}}

\title{A Note on Approximating Weighted Independence on Intersection Graphs of Paths on a Grid}

\author{Saeed Mehrabi}
\affil{\small{School of Computer Science

					Carleton University, Ottawa, Canada

					\url{mehrabi235@gmail.com}
					}
}

\date{}

\begin{document}

\maketitle

\begin{abstract}
A graph $G$ is called \emph{B$_k$-VPG}, for some constant $k\geq 0$, if it has a string representation on an axis-parallel grid such that each vertex is a path with at most $k$ bends and two vertices are adjacent in $G$ if and only if the corresponding paths intersect each other. The part of a path that is between two consecutive bends is called a \emph{segment} of the path.
In this paper, we study the \mwis problem on \bKvpg graphs. The problem is known to be \textsc{NP}-complete on \bvpg graphs, even when the two segments of every path have unit length~\cite{LahiriMS15}, and $O(\log n)$-approximation algorithms are known on \bKvpg graphs, for $k\leq 2$~\cite{DerkaB16,arXivMehrabi2017}. In this paper, we give a $(ck+c+1)$-approximation algorithm for the problem on \bKvpg graphs for any $k\geq 0$, where $c>0$ is the length of the longest segment among all segments of paths in the graph. Notice that $c$ is not required to be a constant; for instance, when $c\in O(\log \log n)$, we get an $O(\log \log n)$-approximation or we get an $O(1)$-approximation when $c$ is a constant. To our knowledge, this is the first $o(\log n)$-approximation algorithm for a non-trivial subclass of \bKvpg graphs.
\end{abstract}

\section{Introduction}
\label{sec:introduction}
In this paper, we study the \mwis problem on \bKvpg graphs. A graph is said to have a \emph{Vertex intersection of Paths in a Grid} (VPG representation, for short), if its vertices can be represented as simple paths on an axis-parallel grid such that two vertices are adjacent if and only if the corresponding paths share at least one grid node. Although VPG graphs were considered a while ago when studying string graphs~\cite{MiddendorfP92}, they were formally investigated by Asinowski et al.~\cite{AsinowskiCGLLS12}. Since then much of the work has focused on studying subclasses of VPG graphs; in particular, restricting the type of paths that are allowed. A turn of a path at a grid node is called a \emph{bend} and a VPG graph is called \emph{B$_k$-VPG graph} if every path has at most $k$ bends.
\begin{definition}[\bKvpg Graph]
\label{def:B1VPG}
A graph $G=(V, E)$ is called a \emph{\bKvpg graph}, if every vertex $u$ of $G$ can be represented as a \emph{path} $P_u$ on an axis-parallel grid $\mathcal{G}$ such that \begin{inparaenum}[(i)] \item $P_u$ has at most $k$ bends, and \item two paths $P_u$ and $P_v$ intersect each other at a grid node if and only if $(u, v)\in E$. \end{inparaenum}
\end{definition}

We remark that by intersecting each other in Definition~\ref{def:B1VPG}, we \emph{include} the case where two paths only \emph{touch} each other. Figure~\ref{fig:sampleGraphs} shows a graph with its \bKvpg representations for $k=1$.

\begin{figure}[t]
\centering
\includegraphics[width=0.70\textwidth]{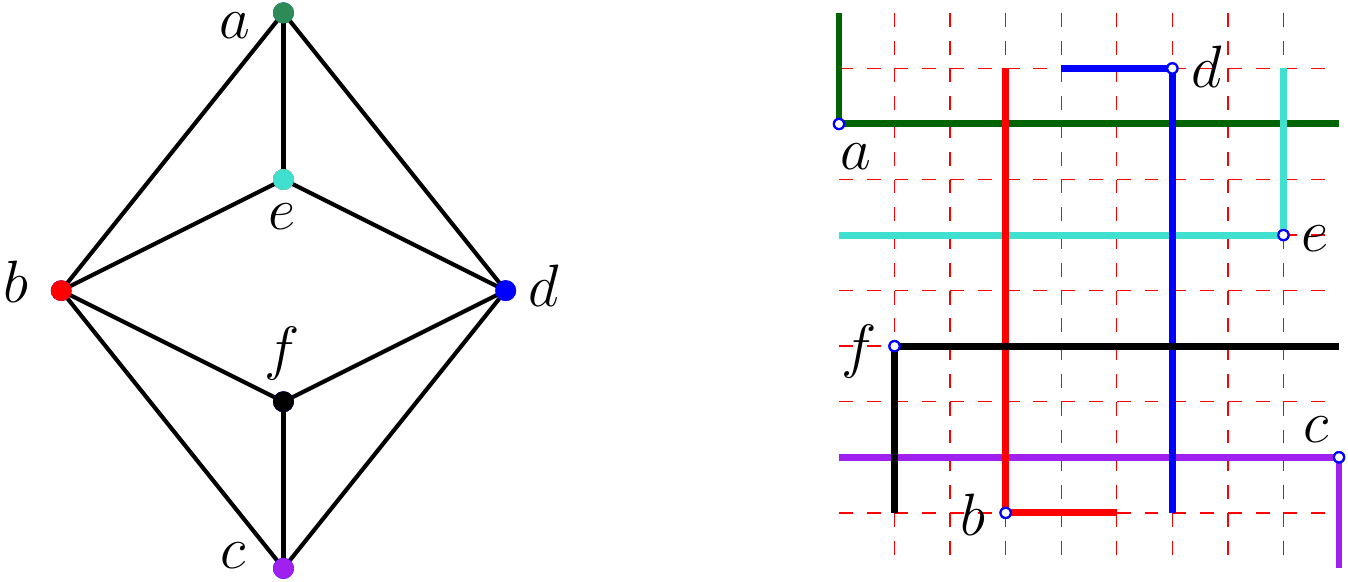}
\caption{A graph on six vertices (left) with its \bKvpg (right) representation for $k=1$.}
\label{fig:sampleGraphs}
\end{figure}

Let $G=(V, E)$ be an undirected graph. A set $S\subseteq V$ is an independent set if no two vertices in $S$ are adjacent. In the \mwis problem, each vertex is assigned a weight and the objective is to compute an independent set of $G$ whose weight is maximum over all independent sets in $G$. The weight of an independent set is defined as the sum of the weights of vertices in the independent set. The \mwis is well known to be \textsc{NP}-hard, and $n^{1-\epsilon}$-hard to approximate for any $\epsilon>0$ unless \textsc{NP}=\textsc{ZPP}~\cite{Hastad01}.

\paragraph{Related work.} Much of the early work is focused on deciding the existence or recognition of VPG representations of graphs~\cite{EhrlichET76,GolumbicLS09,ChalopinGO10,ChaplickU13}. Recently, however, there are works done on designing approximation algorithms for optimization problems on \bKvpg graphs. For instance, \mds problem is \apx on one-string \bvpg graphs\footnote{A \bvpg graph is called \emph{one-string} if every pair of strings that are adjacent in the graph intersect each other exactly once.} due to the facts that every circle graph is a one-string \bvpg graph~\cite{AsinowskiCGLLS12} and that \mds is \apx on circle graphs~\cite{DamianP06}. Moreover, Mehrabi~\cite{waoaPaper2017} studied the \mds problem on some subclass of \bvpg graphs and gave an $O(1)$-approximation algorithm for this problem (see also~\cite{arXivMehrabi2017}). For the \mwis problem, the decision version was shown to be \textsc{NP}-complete on \bvpg graphs by Lahiri et al.~\cite{LahiriMS15} who also gave an $O((\log n)^2)$-approximation algorithm for this problem; notice that the best known algorithm for \mwis on arbitrary string graphs has an approximation factor $n^{\epsilon}$, for some $\epsilon>0$~\cite{FoxP11}. Moreover, $O(\log n)$-approximation algorithms exist for the problem on \bKvpg graphs, for $k\leq 2$~\cite{DerkaB16,arXivMehrabi2017}. It has been asked several times~\cite{LahiriMS15,DerkaB16,arXivMehrabi2017} whether there exist a constant-factor (or even an $o(\log n)$-factor) approximation algorithm for the \mwis problem on \bKvpg graphs.

\paragraph{Our result.} As the first step towards obtaining $o(\log n)$-approximations, we consider the problem on \bKvpg graphs for which the longest segment among all segments of paths in the graph has length $c$, for some $c>0$. We give a polynomial-time $(ck+c+1)$-approximation algorithm for the problem on such graphs (Section~\ref{sec:mWIS}). Notice that $c$ is not required to be a constant; for instance, when $c\in O(\log \log n)$, we get an $O(\log \log n)$-approximation or we get an $O(1)$-approximation when $c$ is a constant. To our knowledge, this is the first $o(\log n)$-approximation algorithm for a non-trivial subclass of \bKvpg graphs.

Our algorithm is based on a linear programming formulation of the problem and then applying the \emph{local ratio} technique~\cite{BarYehudaE85} to the weight function with respect to a feasible solution of the linear program. The algorithm uses a rounding lemma, which exploits the properties of \bKvpg graphs. This will then be used to decompose the weight function for subsequent recursive steps of the algorithm.

\section{Preliminaries}
\label{sec:prelimins}
We denote the string representation of a \bKvpg graph $G=(V, E)$ by $\langle$\piviK$,\mathcal{G}\rangle$, where \piviK is the collection of paths corresponding to the vertices of $G$ and $\mathcal{G}$ is the underlying grid. We might sometimes violate the wording and say \emph{path(s) in $G$} to actually refer to the vertices in $G$ corresponding to the paths in \piviK. Since the recognition problem is \textsc{NP}-hard on such graphs~\cite{PergelR16,ChaplickJKV12}, we assume throughout this paper that a string representation of a \bKvpg graph is always given as part of the input (in addition to $G$). We denote the path in \piviK corresponding to a vertex $u\in V$ by $P_u$. We say that two paths $P_u$ and $P_v$ are adjacent if the vertices $u$ and $v$ are adjacent in $G$. Let $P$ be a path in \piviK. We call the common endpoint of two consecutive segments of $P$ a \emph{corner} of $P$ and denote it by $\cor(P)$. Let $N[P]$ denote the set of paths adjacent to $P$; we assume that $P\in N[P]$. Moreover, for a set $S\subseteq$\piviK of paths, define $N[S]:=\cup_{P\in S}N[P]$. Finally, we denote the $x$- and $y$-coordinates of a point $p$ in the plane by $x(p)$ and $y(p)$, respectively.

Let $\vec{w}\in \mathbb{R}^n$ be a weight vector, and let $F$ be a set of feasibility constraints on vectors $\vec{x}\in \mathbb{R}^n$. A vector $\vec{x}\in\mathbb{R}^n$ is a feasible solution to a given problem $(F, \vec{p})$ if it satisfies all of the constraints in $F$. The value of a feasible solution $\vec{x}$ is the inner product $\vec{w}\cdot\vec{x}$. A feasible solution is \emph{optimal} for a maximization (resp., minimization) problem if its value is maximal (resp., minimal) among all feasible solutions. A feasible solution $\vec{x}$ for a maximization (resp., minimization) problem is an $\alpha$-approximation solution, or simply an $\alpha$-approximation, if $\vec{w}\cdot\vec{x}\geq \alpha\cdot\vec{w}\cdot\vec{x}^*$ (resp., if $\vec{w}\cdot\vec{x}\leq \alpha\cdot\vec{w}\cdot\vec{x}^*$), where $\vec{x}^*$ is an optimal solution. An algorithm is said to have an \emph{approximation factor} of $\alpha$ if it always computes $\alpha$-approximation solutions.

Our $(ck+c+1)$-approximation algorithm for the \mwis problem on \bKvpg graphs is based on the \emph{local ratio} technique, which was first developed by Bar-Yehuda and Even~\cite{BarYehudaE85}. Let us formally state the local ratio theorem.
\begin{theorem}(Local Ratio~\cite{BarYehudaE85})
\label{thm:localRatio}
Let $F$ be a set of constraints, and let $\vec{w}, \vec{w}_1$ and $\vec{w}_2$ be weight vectors where $\vec{w}=\vec{w}_1+\vec{w}_2$. If $\vec{x}$ is an $\alpha$-approximation solution with respect to $(F, \vec{w}_1)$ and with respect to $(F, \vec{w}_2)$, then $\vec{x}$ is an $\alpha$-approximation solution with respect to $(F, \vec{w})$.
\end{theorem}

We now describe how the local ratio technique is usually used for solving a problem. First, the solution set is empty. The idea is to find a decomposition of the weight vector $\vec{w}$ into $\vec{w}_1$ and $\vec{w}_2$ such that $\vec{w}_1$ is an ``easy'' weight function in some sense (we will discuss this in more details later). The local ratio algorithm continues recursively on the instance $(F, \vec{w}_2)$. We assume inductively that the solution returned recursively for the instance $(F, \vec{w}_2)$ is a good approximation and will then prove that it is also a good approximation for $(F, \vec{w})$. This requires proving that the solution returned recursively for the instance $(F, \vec{w}_2)$ is also a good approximation for the instance $(F, \vec{w}_1)$. This step is usually the main part of the proof of the approximation factor.

\section{Approximation Algorithm}
\label{sec:mWIS}
In this section, we give a $(ck+c+1)$-approximation algorithm for the \mwis problem on \bKvpg graphs, for any $k\geq 0$, where $c>0$ is the length of the longest segment among all segments of paths in the graph; notice that $c$ is not required to be a constant. The algorithm is based on rounding a fractional solution derived from a linear programming relaxation of the problem. The standard linear programming relaxation of the \mwis problem is as follows. For each path $P\in\mathcal{P}$, we define an indicator variable $x(P)$. The path $P$ is in the independent set if and only if $x(P)=1$. The integer program assigns the binary values to the paths with the constraint that for each clique $\mathcal{Q}$, the sum of the values assigned to all paths in $\mathcal{Q}$ is at most 1. By relaxing the integer constraint, we get the following linear program in which $w(P)$ denotes the weight of $P$.
\begin{align}
\label{aln:LP}
\text{maximize }        & \sum_{P\in\mathcal{P}}w(P)\cdot x(P)\\
\nonumber \text{subject to }      & \sum_{P\in \mathcal{Q}}x(P)\leq 1 & \forall \mbox{ cliques } \mathcal{Q}\in G,\\
\nonumber                         & x(P)\geq 0 & \forall P\in\mathcal{P}.
\end{align}

Let $\vec{x}$ denote the indicator vector of variables. Notice that any independent set in $G$ gives a feasible integral solution to the linear program. Therefore, the value of an optimal (not necessarily integer) solution to the linear program is an upper bound on the value of an optimal integral solution. The linear program~\eqref{aln:LP} might not have a polynomial number of constraints in general, however, we show that having a polynomial number of constraints in~\eqref{aln:LP} can be guaranteed for \bKvpg graphs. We first need some definitions. For a path $P\in\mathcal{P}$, let $\gr(P)$ be the set of grid points on which $P$ lies; that is, a grid point $t$ is in $\gr(P)$ if and only if $P$ contains $t$. Moreover, we define
\[
\gr(G):=\bigcup_{P\in\mathcal{P}}\gr(P).
\]

Consider any clique $\mathcal{Q}$ in $G$. Observe that the paths in $\mathcal{Q}$ must have at least one grid point in common; that is,
\[
\bigcap_{P\in\mathcal{Q}}\gr(P)\ne\emptyset.
\]
This means that every clique in $G$ is defined by some grid point. For a grid point $t$, let $\pat(t)$ denote the set of all paths in $G$ that contain $t$. Then, we can re-formulate~\eqref{aln:LP} as follows:
\begin{align}
\label{aln:polyLP}
\text{maximize }        & \sum_{P\in\mathcal{P}}w(P)\cdot x(P)\\
\nonumber \text{subject to }      & \sum_{P\in \pat(t)}x(P)\leq 1 & \forall \mbox{ points } t\in \gr(G),\\
\nonumber                         & x(P)\geq 0 & \forall P\in\mathcal{P}.
\end{align}

Since each path $P$ has at most $k+1$ segments and each segment has length at most $c$, we have $|\gr(P)|\leq ck+c+1$ and so $|\gr(G)|\leq (ck+c+1)n$. As such, the number of constraints of~\eqref{aln:polyLP} is polynomial in $n$. Therefore, an optimal solution to~\eqref{aln:polyLP} can be computed in polynomial time. The key to our rounding algorithm is the following lemma.
\begin{lemma}
\label{lem:oneExists}
Let $\vec{x}$ be a feasible solution to the program~\eqref{aln:LP}. Then, for any path $P\in\mathcal{P}$
\[
\sum_{P'\in N[P]}x(P')\leq ck+c+1.
\]
\end{lemma}
\begin{proof}
Take any path $P$ in $\mathcal{P}$. Consider the neighbours of $P$ in $\mathcal{P}$ and partition them into at most $ck+c+1$ sets based on the grid point they share with $P$. If a path $P'\in N[P]$ shares more than one grid point with $P$, then assign it to (only) one of the partition sets arbitrarily. Consider a grid point $t\in \gr(P)$. We know by~\eqref{aln:polyLP} that
\[
\sum_{P'\in\pat(t)}x(P')\leq 1.
\]
Therefore,
\[
\sum_{P'\in N[P]}x(P')=\sum_{t\in\gr(P)}\sum_{P'\in\pat(t)}x(P')\leq\sum_{t\in\gr(P)}(1)\leq ck+c+1.
\]
\end{proof}

By Lemma~\ref{lem:oneExists}, the rounding algorithm applies a local ratio decomposition of the weight vector $\vec{w}$ with respect to $\vec{x}$. See Algorithm~\ref{alg:approxMWISAlg}. Clearly, the set $S$ returned by the algorithm is an independent set. The following lemma establishes the approximation factor of the algorithm.

\begin{algorithm}[t]
\caption{\textsc{ApproximateMWISon\bKvpg}($\langle$\piviK$,\mathcal{G}\rangle$)}
\label{alg:approxMWISAlg}
\begin{algorithmic}[1]
\State Delete all paths with non-positive weight. If no paths remain, then return the empty set.
\State Let $P\in$\piviK be a path satisfying
\[
\sum_{P'\in N[P]}x(P)\leq ck+c+1.
\]
Then, decompose $\vec{w}$ into $\vec{w}:=\vec{w_1}+\vec{w_2}$ as follows:
\[   
w_1(P') = 
     \begin{cases}
       w(P) &\quad\text{if } P'\in N[P],\\
       0 &\quad\text{otherwise.}\\
     \end{cases}
\]
\State Solve the problem recursively using $\vec{w_2}$ as the weight vector. Let $S'$ be the independent set returned.
\State If there exists at least one path in $S'$ that is adjacent to $P$, then return $S:=S'$; otherwise, return $S:=S'\cup P$.
\end{algorithmic}
\end{algorithm}

\begin{lemma}
\label{lem:approxFactor}
Let $\vec{x}$ be a feasible solution to~\eqref{aln:polyLP}. Then, $w(S)\geq \frac{1}{ck+c+1}(\vec{w}\cdot\vec{x})$.
\end{lemma}
\begin{proof}
We prove the lemma by induction on the number of recursive calls. In the base case, the set returned by the algorithm satisfies the lemma because no vertices are remained. Moreover, the first step that removes all vertices with non-positive weight cannot decrease the right-hand side of the above inequality. We next prove the induction step.

Suppose that $\vec{z}$ and $\vec{z'}$ correspond to the indicator vectors for $S$ and $S'$, respectively. By induction, $\vec{w_2}\cdot\vec{z'}\geq\frac{1}{ck+c+1}(\vec{w_2}\cdot\vec{x})$. Since $w_2(P)=0$, we have $\vec{w_2}\cdot\vec{z}\geq\frac{1}{ck+c+1}(\vec{w_2}\cdot\vec{x})$. From the last step of the algorithm, we know that at least one path from $N[P]$ is in $S$ (recall that we assumed $P\in N[P]$) and so we have
\[
\vec{w_1}\cdot\vec{z}=w(P)\sum_{P'\in N[P]}z(P')\geq w(P).
\]
Moreover, by Lemma~\ref{lem:oneExists},
\[
\vec{w_1}\cdot\vec{x}=w(P)\sum_{P'\in N[P]}x(P')\leq (ck+c+1)\cdot w(P).
\]
Therefore,
\begin{align*}
& \vec{w_1}\cdot\vec{x}\leq (ck+c+1)\cdot w(P)\leq (ck+c+1)\cdot(\vec{w_1}\cdot\vec{z})\\
& \Rightarrow \vec{w_1}\cdot\vec{z}\geq\frac{1}{ck+c+1}(\vec{w_1}\cdot\vec{x})\\
& \Rightarrow (\vec{w_1}+\vec{w_2})\cdot\vec{z}\geq\frac{1}{ck+c+1}(\vec{w_1}+\vec{w_2})\cdot\vec{x}\\
& \Rightarrow w(S)\geq\frac{1}{ck+c+1}\vec{w}\cdot\vec{x}.
\end{align*}
This completes the proof of the lemma.
\end{proof}

Since there exists at least one path $P$ for which $w_2(P)=0$ in each recursive step, Algorithm~\ref{alg:approxMWISAlg} terminates in polynomial time. Therefore, by Lemmas~\ref{lem:approxFactor}, we have the main result of this paper.
\begin{theorem}
\label{thm:approxMWIS}
There exists a polynomial-time $(ck+c+1)$-approximation algorithm for the \mwis~problem on \bKvpg graphs, for any $k>0$, where $c>0$ is the length of the longest segment in the graph.
\end{theorem}

\section{Conclusion}
\label{sec:conclusion}
In this paper, we considered the \mwis problem on \bKvpg graphs. We gave a polynomial-time $(ck+c+1)$-approximation algorithm on \bKvpg graphs, for any $k>0$, where $c$ denotes the length of the longest segment of all paths in the graph. The algorithm relies crucially on the fact that the longest segment in the graph has length $c$. Designing an $o(\log n)$-approximation of a maximum-weighted independent set on any \bKvpg graph (even for $k=1$) remains open.

\bibliographystyle{plain}
\bibliography{ref}

\end{document}